%% file: paper.tex
\documentclass[twoside,11pt,final]{entics}
\usepackage{enticsmacro}
\usepackage{graphicx}
\usepackage[all]{xy}

\input{setup}
\input{macros}

\def\conf{MFPS 2025} 	
\volume{5}			
\def\volu{5}			
\def\papno{7}			
\def\mydoi{16512}
\def\lastname{Van Collem, Van der Weide, Geuvers} 
\def\runtitle{Initial Algebras of Domains via QIITs} 
\def\copynames{S. van Collem, N. van der Weide, H. Geuvers}    
\begin{document}
\begin{frontmatter}
  \title{Initial Algebras of Domains via Quotient Inductive-Inductive Types} 						
  \author{Simcha van Collem\thanksref{a}\thanksref{simcha-email}\thanksref{erc}}	
   \author{Niels van der Weide\thanksref{a}\thanksref{niels-email}\thanksref{nwo}}		
   \author{Herman Geuvers\thanksref{a}\thanksref{b}\thanksref{herman-email}}		
   \address[a]{Institute for Computing and Information Sciences\\ Radboud University\\				
    Nijmegen, The Netherlands}
   \address[b]{Faculty of Mathematics and Computer Science\\ Technical University Eindhoven\\				
    The Netherlands}
   \thanks[simcha-email]{Email: \href{mailto:simcha.vancollem@ru.nl} {\texttt{\normalshape
        simcha.vancollem@ru.nl}}}
  \thanks[niels-email]{Email:  \href{mailto:nweide@cs.ru.nl} {\texttt{\normalshape
        nweide@cs.ru.nl}}}
  \thanks[herman-email]{Email:  \href{mailto:herman@cs.ru.nl} {\texttt{\normalshape
        herman@cs.ru.nl}}}
  \thanks[erc]{The first author was supported by ERC grant COCONUT (grant agreement no. 101171349), funded by the European Union. Views and opinions expressed are however those of the author(s) only and do not necessarily reflect those of the European Union or the European Research Council Executive Agency. Neither the European Union nor the granting authority can be held responsible for them.}
  \thanks[nwo]{The second author was supported by the NWO project ``The Power of Equality'' OCENW.M20.380, which is financed by the Dutch Research Council (NWO).}
  \begin{abstract}
    Domain theory has been developed as a mathematical theory of
    computation and to give a denotational semantics to programming
    languages. It helps us to fix the meaning of language concepts, to
    understand how programs behave and to reason about programs.  At
    the same time it serves as a great theory to model various
    algebraic effects such as non-determinism, partial functions, side
    effects and numerous other forms of computation.

    In the present paper, we present a general framework to construct
    algebraic effects in domain theory, where our domains are DCPOs:
    directed complete partial orders. We first describe so called {\em
      DCPO algebras\/} for a signature, where the signature specifies
    the operations on the DCPO and the inequational theory they
    obey. This provides a method to represent various algebraic
    effects, like partiality. We then show that initial DCPO algebras
    exist by defining them as so called {\em Quotient
      Inductive-Inductive Types} (QIITs), known from homotopy type
    theory. A quotient inductive-inductive type allows one to
    simultaneously define an inductive type and an inductive relation
    on that type, together with equations on the type. We illustrate
    our approach by showing that several well-known constructions of
    DCPOs fit our framework: coalesced sums, smash products and free
    DCPOs (partiality and power domains).  Our work makes use of
    various features of homotopy type theory and is formalized in
    Cubical Agda.
%
%
  \end{abstract}
\begin{keyword}
  domain theory, quotient inductive-inductive types, algebra, algebraic effects.
\end{keyword}
\end{frontmatter}

\section{Introduction}

Domain theory was developed in the late 1960s by Dana
Scott~\cite{Scott70} as a mathematical theory of computation and to
give a {\em denotational semantics\/} to programs and programming
languages, one that abstracts away from the operational semantics of a
program, and describes the semantics of a program in terms of the
mathematical function it denotes.  Domain theory allows one to give a
high level meaning to programs and program constructions and to fix
the meaning of concepts from programming languages. It also allows to
reason about programs on a higher level of abstraction.
Since the 60s, programming languages have evolved, introducing many new programming paradigms.
To keep up with the rapid development of these new constructs in programming languages,
and to fully understand how they work,
we also need to be able to describe these using denotational semantics.

Algebraic effects~\cite{lmcs:705} have been introduced to represent computational effects in programming, such as
state, exceptions, nondeterminism, non-termination, input-output, and many more.
They allow factoring out effectful computations from pure computations and can be composed easily.
An effect is defined as a set of operations and an (in)equational theory these operations should obey.
For example, non-termination has two operations; one operation represents the non-termination,
while the other represents the possible returned value.

In domain theory these algebraic effects are a \emph{directed complete partial order} (DCPO) $D$ together with operations on $D$ which obey an inequational
theory,
and we call these DCPO algebras.
Of particular importance is the \emph{initial DCPO algebra}.
For example, for partiality, the inequational theory states
that the non-termination operation should be smaller than all possible
return values, as non-termination gives the least amount of
information about the return value of a partial function.
Another example is given by the \emph{powerdomain}~\cite{Plotkin76}.
Its inequational theory is similar to that of a join-semilattice,
and it is used to represent nondeterminism.

There are various constructions of initial DCPO algebras~\cite{JUNG2008209}.
However, these constructions use power sets,
and thus they are only suitable in impredicative foundations.
In this paper, we present an alternative construction for initial DCPO algebras.
The main idea is to use \emph{quotient inductive-inductive types} (QIITs)~\cite{DBLP:conf/fossacs/AltenkirchCDKF18}.
Using QIITs, we can define a DCPO by specifying operations and inequalities.
This idea was also used by Altenkirch, Danielsson, and Kraus~\cite{Partiality} to construct the partiality monad.
We use this mechanism to show that one can define initial algebras for a suitable notion of signature.
Note that this construction is predicative, because it does not use power sets.

\subsection{Contribution and Overview}
In \cite{Partiality} a quotient inductive-inductive type (QIIT, ~\cite{DBLP:conf/fossacs/AltenkirchCDKF18}) is used
to construct the non-termination effect in domain theory.
A quotient inductive-inductive type is a construction from homotopy type theory that allows one to simultaneously define an inductive type with a quotient on it, in parallel with an inductively defined relation on this type.
In the present paper, we extend the method of \cite{Partiality} and present DCPO algebras as a general framework for constructing algebraic effects in domain theory.

We start this paper by recalling quotient inductive-inductive types in Section~\ref{sec:prelims-qiits}.
After that we describe and compare several methods that one can use to construct DCPO algebras in Section~\ref{sec:constructing-cpos}.
We also discuss the main idea of our construction in that section.
In Sections~\ref{sec:signatures} and~\ref{sec:algebras}, we define a notion of signature and algebra for it,
and we illustrate these notions using the Plotkin powerdomain.
We construct an initial algebra for every signature in Section~\ref{sec:initial-algebra},
and we present numerous examples in Section~\ref{sec:examples}.
Finally, we discuss related work in Section~\ref{sec:related-work},
and we conclude in Section~\ref{sec:conclusion}.


\subsection{Formalization}
As a foundation, our work uses homotopy type theory (HoTT).
In particular, we use QIITs and some of the arguments make use of function extensionality.
We do not need to use the univalence axiom in our work.
Our work is formalized using Cubical Agda~\cite{cubical}.
The artifact~\cite{artifact} includes a file with references to the Agda code for all the definitions and theorems.%
\footnote{A clickable html version can be found here: \href{https://simchavc.github.io/artifacts/qiits/Paper.html}{https://simchavc.github.io/artifacts/qiits/Paper.html}.}

\section{Preliminaries on Quotient Inductive-Inductive Types}
\label{sec:prelims-qiits}
We start this paper by recalling \emph{quotient inductive-inductive types} (QIIT)~\cite{DBLP:conf/fossacs/AltenkirchCDKF18,kaposi:2020a}.
In essence, quotient inductive-inductive types combine the ideas of \emph{inductive-inductive types} (IIT)~\cite{nordvallforsberg2013thesis}
and \emph{quotient types} (QIT)~\cite{univalentfoundationsprogram:2013}.
Inductive-inductive types are used to define a type $A$ together with a type family on $A$,
and quotient inductive types are used to define types by specifying constructors and equations between them.

To get an understanding of what quotient inductive-inductive types are,
we consider an example,
and we refer the reader for a formal definition to the literature~\cite{DBLP:conf/fossacs/AltenkirchCDKF18}.
In addition,
we only look at how to specify QIITs
and we do not give their elimination rules.
Let us start with the type of ``finite sets of natural numbers'', which can be represented as a QIT in various ways, see \cite{Fruminetal2018}.
Its constructors are given in Figure~\ref{fig:qit-fin-set}.
The $\emptyList$ and $\cons{x}{xs}$ constructors are the same as one would expect for regular lists.
Apart from the regular point constructors, we also have path constructors, which states that $\finSet$ actually is the type of finite sets of natural numbers.
First, \nameref{rule:fin-set-cons-comm} makes sure that the order of the elements in a finite set does not matter.
Secondly, \nameref{rule:fin-set-cons-dup} says that adding an element is idempotent.
Combining these two path constructors gives us the structure we would expect.
However, by introducing these path constructor, we now have multiple paths between finite sets.
For example, we can identify the list $\cons{x}{\cons{x}{\emptyList}}$ with $\cons{x}{\emptyList}$ either by using \nameref{rule:fin-set-cons-dup},
or by the composition of \nameref{rule:fin-set-cons-dup} and \nameref{rule:fin-set-cons-comm}.
These paths are not equal to each other.
This introduces extra structure which is not present when we consider finite multisets.
We therefore add the path constructor \nameref{rule:fin-set-set}.
As this path constructor asserts that $\finSet$ is a set,
we typically write this by adding a rule with the conclusion $\isSet(\finSet)$.

\begin{figure}
    \begin{mathpar}
        \inferrule
            { }
            {\emptyList : \finSet}

        \inferrule
            {x : \Nat \and xs : \finSet}
            {\cons{x}{xs} : \finSet}
        \\
        \inferrule
            [\textsc{Cons-Comm}]
            {x, y : \Nat \and xs : \finSet}
            {\cons{x}{\cons{y}{xs}} = \cons{y}{\cons{x}{xs}}}
            [\textsc{Cons-Comm}]\label{rule:fin-set-cons-comm}

        \inferrule
            [\textsc{Cons-Dup}]
            {x : \Nat \and xs : \finSet}
            {\cons{x}{\cons{x}{xs}} = \cons{x}{xs}}
            [\textsc{Cons-Dup}]\label{rule:fin-set-cons-dup}
        \\
        \inferrule
            [FinSet-Set]
            {xs, ys : \finSet \and p, q : xs = ys}
            {p = q}
            [\textsc{FinSet-Set}]\label{rule:fin-set-set}
    \end{mathpar}
    \caption{Constructors for $\finSet$}\label{fig:qit-fin-set}
\end{figure}

Next up, we shift our focus to an IIT.
As an example, we consider the type of sorted lists~\cite[Example~3.2]{nordvallforsberg2013thesis}.
Its constructors are given in Figure~\ref{fig:iit-sorted-list}.
The predicate $x \leqList xs$ asserts that $x$ is smaller than all elements of the sorted list $xs$.
This predicate allows us to define the type $\sortedList$.
We have that $\emptyList$ is sorted, and if $xs$ is sorted and we have a proof $p$ of the fact that $x \leqList xs$,
then we have that $\consSorted{x}{p}{xs}$ is again a sorted list.
To define the predicate, we know that $x$ is always smaller than all elements of the empty list.
Furthermore, if $p$ is a proof of the fact that $x \leqList xs$,
then from $n \leq x$ we can conclude that $n$ is also smaller than all elements of $\consSorted{x}{p}{xs}$.

\begin{figure}
    \begin{mathpar}
        \inferrule
            { }
            {\emptyList : \sortedList}

        \inferrule
            {x : \Nat \and xs : \sortedList \and p : x \leqList xs}
            {\consSorted{x}{p}{xs} : \sortedList}
        \\
        \inferrule
            {x : \Nat}
            {x \leqList \emptyList}

        \inferrule
            {n, x : \Nat \and xs : \sortedList \and p : x \leqList xs}
            {n \leq x \to n \leqList \consSorted{x}{p}{xs}}
    \end{mathpar}
    \caption{Constructors for $\sortedList$}\label{fig:iit-sorted-list}
\end{figure}

Finally, we combine the previous two example into a QIIT.
This gives us sorted lists where any two consecutive elements are unequal.
Its constructors are given in Figure~\ref{fig:qiit-sorted-list}.
We again have the $\emptyList$ and $\consSorted{x}{p}{xs}$ constructors to create sorted lists
and corresponding constructors to show that these lists are sorted.
We now also add a constructor which states that $\consSorted{x}{q}{\consSorted{x}{p}{xs}}$ and $\consSorted{x}{p}{xs}$ are equal.
Again, like in the first example, we also add a set truncation constructor to remove the extra structure we are not interested in.

\begin{figure}
    \begin{mathpar}
        \inferrule
            { }
            {\emptyList : \strictSortedList}

        \inferrule
            {x : \Nat \and xs : \strictSortedList \and p : x \leqList xs}
            {\consSorted{x}{p}{xs} : \strictSortedList}
        \\
        \inferrule
            {x : \Nat \and xs : \strictSortedList \and p : x \leqList xs \and q : x \leqList \consSorted{x}{p}{xs}}
            {\consSorted{x}{q}{\consSorted{x}{p}{xs}} = \consSorted{x}{p}{xs}}

        \inferrule
            { }
            {\isSet(\strictSortedList)}
        \\
        \inferrule
            {x : \Nat}
            {x \leqList \emptyList}

        \inferrule
            {n, x : \Nat \and xs : \strictSortedList \and p : x \leqList xs}
            {n \leq x \to n \leqList \consSorted{x}{p}{xs}}
    \end{mathpar}
    \caption{Constructors for $\strictSortedList$}\label{fig:qiit-sorted-list}
\end{figure}

\section{Constructions of DCPOs}
\label{sec:constructing-cpos}
In this section,
we discuss two constructions of \emph{directed complete partial order} (DCPO)
that can be used for DCPO algebras.
The first way is given by \emph{presentations} of DCPOs~\cite{JUNG2008209}.
Intuitively, a presentation specifies a basis and relations on elements on that basis.
From a presentation one can construct a DCPO
in a way similar to the rounded ideal completion~\cite{abramskyJung}.
The other construction is by using quotient inductive-inductive types, following ideas by Altenkirch, Danielsson, and Kraus~\cite{DBLP:conf/fossacs/AltenkirchCDKF18}.

Let us start by explaining why we are interested in DCPOs rather than $\omega$-CPOs.
An $\omega$-CPO is a partial order where every increasing sequence $(d_i)_{i \in \mathbb{N}}$,
has a least upper bound.
The notion of DCPO is stronger: here we require that every \emph{directed} family has a least upper bound.
Recall that a family $\alpha : I \to D$ is directed if $I$ is inhabited,
and for all $i, j : I$, there exists a $k : I$ such that $\alpha(i), \alpha(j) \sqsubseteq \alpha(k)$.
Concretely, a DCPO is given by a partially ordered set $D$ such that every directed family has a least upper bound.
While every DCPO is also an $\omega$-CPO,
proving that every $\omega$-CPO is a DCPO requires the index types $I$ to be countable and the axiom of choice~\cite{Markowsky1976ChaincompletePA}.

If one works predicatively with DCPOs,
then one has to keep careful track of the universe levels involved~\cite{dejong:2021,dejong:2021b}.
The type of DCPOs is indexed by three universe levels.
More specifically,
we have a type $\DCPO_{\Universe_1, \Universe_2, \Universe_3}$
that consists of DCPOs $D$
where the carrier lives in $\Universe_1$
and the order $\sqsubseteq$ is valued in $\Universe_2$ (we have ${ \sqsubseteq } : D \rightarrow D \rightarrow \Universe_2$)
and such that every directed family indexed by a type $I : \Universe_3$
has a least upper bound.
In the formalization the universe levels are explicitly recorded, but in this paper we leave them implicit.

\subsection{The Rounded Ideal Completion and Presentations}
We start by discussing presentations, and to introduce these, we first discuss the \emph{rounded ideal completion}.
The rounded ideal completion gives us a way to construct DCPOs by specifying an \emph{abstract basis}~\cite{abramskyJung}.
Elements of the basis are the basic approximations for the elements in the DCPO that we construct.

\begin{definition}
    \label{def:abstract-basis}
    A pair $(B, \prec)$, where $\prec {}\colon B \to B \to \Omega$, is an \textbf{abstract basis} if $\prec$ is transitive,
    and it has the following interpolation properties:
    \begin{itemize}
        \item if $x : B$, then there exists $y : B$ such that $y \prec x$;
        \item if $x_1, x_2, z : B$ such that $x_1, x_2 \prec z$, then there exists $y : B$ such that $x_1, x_2 \prec y \prec z$.
    \end{itemize}
\end{definition}

An instance of an abstract basis is given by the ordered rational numbers~\cite[Definition 106]{dejong:2023}.
Formally, we define the set $B$ to be the collection of rational numbers,
and we say that $p \prec q$ if $p < q$.
The first condition says that for every $p : \mathbb{Q}$ we have $r$ such that $p < r$,
and the second condition says that whenever we have rational numbers $p_1, p_2, q : \mathbb{Q}$
such that $p_1 < q$ and $p_2 < q$,
we have another rational number $r$ such that $p_1 , p_2 < r < q$.

We now consider rounded ideals of such an abstract basis $(B, \prec)$.
These ideals form the underlying type of the rounded ideal completion of $B$.
Following the intuition that these ideals should approximate information,
they are defined as directed lower-subsets of $B$.

\begin{definition}
    Let $(B, \prec)$ be an abstract basis. A predicate $X \colon B \to \Omega$%
    \footnote{Note that we think of $X$ as subset $X \subseteq B$.}
    is a \textbf{rounded ideal} if
    $X$ is an directed lower set.
    The type of rounded ideals, $\RoundedIdeal{B, \prec}$,
    forms a DCPO with subset inclusion as the ordering and union of subsets as the supremum.
\end{definition}

The rounded ideal completion allows us to construct continuous DCPOs by specifying a basis.
For instance, we can construct the Cantor and Baire space as the rounded ideal completion of lists ordered by the prefix relation \cite[Definition 87]{dejong:2023},
and intervals of real numbers as the rounded ideal completion of rational intervals ordered by reverse strict inclusion~\cite{dejong:2023,vanderWeide2024}.

Note that an abstract basis only specifies basis elements and an order on them.
We are thus unable to specify inequalities between basis elements
and suprema of collections of basis elements, which means we cannot require Scott continuity.
For this reason, we consider the more general notion of \emph{DCPO presentations}~\cite{JUNG2008209}.

\begin{definition}\label{def:presentation}
A \textbf{DCPO presentation} consists of a preorder $P$
and a binary relation $\covered$ on $P$ and directed subsets of $P$.
\end{definition}

Intuitively, the preorder $P$ contains the generator elements for the DCPO that we present.
Whenever a generator $p$ is covered by a directed subset $U$, written $p \covered U$,
we intuitively think of $p$ being below the supremum of $U$ in the presented DCPO.
Bidlingmaier, Faissole, and Spitters use presentations,
although for $\omega$-CPO rather than DCPOs, to study the semantics of probabilistic programming languages~\cite[Assumption~1]{bidlingmaier:2021}.
They assume the existence of the free $\omega$-CPO as an axiom, and they discuss various justifications such as impredicativity and QIITs.

One way to justify the existence of the free DCPO $\overline{P}$ for a presentation,
is via an impredicative construction similar to the rounded ideal completion.
Specifically, we say that $I \subseteq P$ is an \emph{ideal} if it is a lower set,
and for each $U \subseteq I$ such that $p \covered U$, we have $p \in I$.
The type of ideals, $\Ideal{P}$, has a complete lattice structure.
For a subset $M \subseteq P$, we write $\langle M \rangle$ for the least ideal containing $M$.
The presented DCPO $\overline{P}$ is defined to be the least sub-DCPO of $\Ideal{P}$,
containing all $\langle \{ q \mid q \leq p \} \rangle$ for $p : P$.
Note that $\overline{P}$ in general lives in a higher universe than $P$,
unless we work in impredicative foundations.

\subsection{Quotient Inductive-Inductive Types}
\label{sec:QIIT}
\begin{figure}[t]
  \begin{mathpar}
    \inferrule
    {}
    {\bot : A_\bot}

    \inferrule
    {a : A}
    {\eta(a) : A_\bot}

    \inferrule
    {a \sqsubseteq b \and b \sqsubseteq a}
    {a = b}

    \inferrule
    {s : \mathbb{N} \rightarrow A_\bot \and p : \prod_{n : \mathbb{N}} s(n) \sqsubseteq s(n + 1)}
    {\bigsqcup(s, p) : A_\bot}
    \\
    \inferrule
    {}
    {x \sqsubseteq x}

    \inferrule
    {x \sqsubseteq y \and y \sqsubseteq z}
    {x \sqsubseteq z}

    \inferrule
    {}
    {\bot \sqsubseteq x}

    \inferrule
    {n : \mathbb{N}}
    {s(n) \sqsubseteq \bigsqcup(s, p)}

    \inferrule
    {\prod_{n : \mathbb{N}} s(n) \sqsubseteq x}
    {\bigsqcup(s, p) \sqsubseteq x}
\end{mathpar}
\caption{Constructors for $A_\bot$ and $\sqsubseteq$}
\label{fig:partiality}
\end{figure}

Another way to construct DCPOs is by using quotient inductive-inductive types~\cite{DBLP:conf/fossacs/AltenkirchCDKF18}.
The reason why one can use QIITs to construct DCPOs,
is because they allow us to simultaneously define a type $A$ together with a relation $R$ between $A$ and $A$
by specifying constructors for both $A$ and $R$.
Because $A$ and $R$ are defined simultaneously,
the constructors of $A$ and $R$ may refer to each other.

This idea was used by Altenkirch, Danielsson, and Kraus to construct the partiality monad~\cite{Partiality}.
More specifically, they construct the free $\omega$-CPO as a QIIT,
and the constructors of this QIIT are given in Figure~\ref{fig:partiality}.
In the remainder of this paper, we generalize their work to construct DCPO algebras.

We illustrate the idea behind our development using an example.
Instead of constructing the free pointed DCPO as a QIIT,
we show how to construct another example, namely the \emph{powerdomain}.
The powerdomain gives a more interesting QIIT,
because it has recursive operations
and one needs to add constructors to the QIIT expressing the Scott continuity of each operation.
In addition,
one can construct the partiality monad for DCPOs in more elementary ways, see the work by Escard\'o and Knapp \cite[Theorem 1]{escardo:2017} and by De Jong and Escard\'o \cite[Theorem 25]{dejong:2021}.

Given a DCPO $D$, the powerdomain $\plotkinPower{D}$ over $D$ is defined to be the free DCPO generated by constructors
$\powerIncl{\blank} : D \to \plotkinPower{D}$ and $\formalUnion : \plotkinPower{D} \to \plotkinPower{D} \to \plotkinPower{D}$
such that $\formalUnion$ is associative, commutative, and idempotent.
We construct the DCPO $\plotkinPower{D}$ together with its order $\sqsubseteq$ as the quotient inductive-inductive type
with the point and path constructors given in Figures \ref{fig:power-domain-dcpo}, \ref{fig:power-domain-operations} and \ref{fig:power-domain-ineqs}.
We omit the constructors which state that $\powerIncl{\blank}$ and $\formalUnion$ are monotone.
Note that we write down the equations of the powerdomain as inequalities,
because the formalism that we introduce in this paper, is based on inequalities rather than equalities.

There are a couple of interesting aspects to this quotient inductive-inductive type.
First, let us note the key difference with the QIIT by Altenkirch, Danielsson, and Kraus~\cite{Partiality}.
Since we look at DCPOs,
every directed family must have a supremum
rather than only chains indexed by the natural numbers.
For this reason, the constructor $\bigsqcup$ quantifies over all possible directed families, whose index type lives in some universe $\Universe$.
The type $\plotkinPower{D}$ thus lives in a larger universe than $\Universe$,
and, if $D$ also lives in $\Universe^+$, the powerdomain $\plotkinPower{D}$ is a type in the successor universe $\Universe^+$.

The second interesting aspect of the definition of this QIIT is that, in essence,
it is purely specified by describing the constructors and the inequalities.
On top of the constructions and rules for a DCPO (Figure \ref{fig:power-domain-dcpo}), the constructors of the powerdomain and their continuity are specified in Figure \ref{fig:power-domain-operations}, and
the inequalities are specified in Figure \ref{fig:power-domain-ineqs}.
If we were to specify any other algebraic structure on DCPOs,
we would only need to modify the constructors in Figures \ref{fig:power-domain-operations} and \ref{fig:power-domain-ineqs}.
In the remainder, we generalize this construction,
and we show how to use QIITs to construct a wide variety of algebraic structures on DCPOs.

\begin{figure}

  \begin{mathpar}
    \inferrule
        { }
        {x \Leq{} x}

    \inferrule
        {x \Leq{} y \and y \Leq{} z}
        {x \Leq{} z}

    \inferrule
        { }
        {\isProp(x \Leq{} y)}
    \\
    \inferrule
        {x \Leq{} y \and y \Leq{} x}
        {x = y}

    \inferrule
        { }
        {\isSet(\plotkinPower{D})}
    \\
    \initialLub{} : \prod_{\alpha : I \to \plotkinPower{D}} \isDirected{\alpha} \to \plotkinPower{D}
    \\
    \inferrule[\textsc{$\initialLub{}$-Is-Sup-1}]
        {\alpha : I \to \plotkinPower{D} \and \delta : \isDirected{\alpha}}
        {\prod_{i : I} \alpha(i) \Leq{} \initialLub{}(\alpha, \delta)}

    \inferrule[\textsc{$\initialLub{}$-Is-Sup-2}]
        {\alpha : I \to \plotkinPower{D} \and \delta : \isDirected{\alpha}}
        {\prod_{v : \plotkinPower{D}} \isUpperbound(v, \alpha) \to \initialLub{}(\alpha, \delta) \Leq{} v}
  \end{mathpar}
  \caption{Constructors for the DCPO structure of $\plotkinPower{D}$}
  \label{fig:power-domain-dcpo}

  \begin{mathpar}
    \powerIncl{\blank} : D \to \plotkinPower{D}

    \formalUnion : \plotkinPower{D} \to \plotkinPower{D} \to \plotkinPower{D}

    \inferrule
        {\alpha_1, \alpha_2 : I \to \plotkinPower{D} \and \isDirected{\alpha_1} \and \isDirected{\alpha_2}}
        {\initialLub{} \alpha_1 \formalUnion \initialLub{} \alpha_2 = \initialLub{} (\lambda i.\, \alpha_1(i) \formalUnion \alpha_2(i))}

    \inferrule
        {\alpha : I \to D \and \isDirected{\alpha}}
        {\powerIncl{\dcpoSup{D} \alpha} = \initialLub{} (\powerIncl{\blank} \circ \alpha)}
  \end{mathpar}
  \caption{Operations of $\plotkinPower{D}$ and their continuity}
  \label{fig:power-domain-operations}

  \begin{mathpar}
    \inferrule
        { }
        {x \formalUnion y \Leq{} y \formalUnion x}

    \inferrule
        { }
        {(x \formalUnion y) \formalUnion z \Leq{} x \formalUnion (y \formalUnion z)}

    \inferrule
        { }
        {x \formalUnion x \Leq{} x}

    \inferrule
        { }
        {x \Leq{} x \formalUnion x}
  \end{mathpar}
  \caption{Inequalities involving operations of $\plotkinPower{D}$}
  \label{fig:power-domain-ineqs}
\end{figure}

\section{Signatures}
\label{sec:signatures}
In this section we define a notion of signature for DCPO algebras,
and we demonstrate this notion using the Plotkin powertheory for a fixed DCPO $D$~\cite{abramskyJung,Plotkin76}.
Our notion of signature is inspired by universal algebra
where algebraic structures are described by specifying operations and equations.
However, our signatures are specified by operations and inequalities instead.

Let us start by defining how to specify operations.
A single operation is represented using a monomial.


\begin{definition}
    A \textbf{monomial} $\monomial$
    consists of
    an arity $\monomialArity{\monomial} : \Universe$
    and a DCPO $\monomialConst{\monomial}$.
    We define $\interp{\monomial}$ as the functor which sends a DCPO $X$ to the DCPO
    $(\monomialArity{\monomial} \to X) \times \monomialConst{\monomial}$.
    We refer to $\monomialConst{\monomial}$ as the constant of the functor $\interp{\monomial}$.
\end{definition}

A monomial $\monomial$ describes an operation $\interp{\monomial}(X) \to X$ on a DCPO $X$.
The Plotkin powertheory of $D$ has two operations,
namely the inclusion operation $D \to X$ and the formal union operation $X \to X \to X$.
The inclusion operation is represented by a monomial with arity $\bot$ and constant $D$,
and formal union can be represented by a monomial with arity $\Bool$ and the unit DCPO as the constant.

To package multiple operations together, we consider families of monomials,
called \emph{presignatures}.

\begin{definition}
    A \textbf{presignature} $\Sigma : \PreSig$ consists of
    a type of constructor names in $\sigConstrNames{\Sigma} : \Universe$
    and a family of monomials, $\sigMonomialFam{\Sigma} : \sigConstrNames{\Sigma} \to \Monomial$.
    We write $\sigMonomialArity{\Sigma}(a)$ and $\sigMonomialConst{\Sigma}(a)$ for the arity and the constant of $\sigMonomialFam{\Sigma}(a)$ respectively.
\end{definition}

As the Plotkin powertheory of $D$ has two operations,
we define the type of constructor names of its presignature to be a type with exactly two elements: $\inclusionName$ and $\formalUnionName$.
The family of monomials is constructed by sending both constructor names to their respective monomial we defined above.

\begin{remark}\label{remark:signature-w-types-comparison}
    The structure of a presignature looks similar to a container~\cite{abbott:2005}, which are used to define W-types~\cite{Lof84},
    but for each $a : \sigConstrNames{\Sigma}$, there is an additional DCPO $\sigMonomialConst{\Sigma}(a)$.
    For containers this is not needed, as one can introduce a new operation for each constant $c : \sigMonomialConst{\Sigma}(a)$.
    However, this does not take into account the DCPO-structure of $\sigMonomialConst{\Sigma}(a)$.
    We want the maps to be Scott continuous (as will be detailed in Section~\ref{sec:algebras}) and these should also be continuous with respect to the DCPO-structure of $\sigMonomialConst{\Sigma}(a)$.
 For example, the monomial for $\inclusionName$ should represent a Scott continuous map $D \to \plotkinPower{D}$,
    which is different from a family of constants $\powerIncl{d}$ for each $d : D$.
\end{remark}

%

Inequalities are described by their left- and right-hand side,
which are constructed using variables and the operations specified by a presignature.
The notion of a \emph{term} formalizes this concept.

\begin{definition}
    Let $\Sigma$ be a presignature and $V$ a type of variables.
    The type $\Term_{\Sigma, V}$ of \textbf{terms} is inductively generated by the following constructors.
    \begin{mathpar}
        \inferrule
            {x : V}
            {\var{x} : \Term_{\Sigma, V}}

        \inferrule
            {a : \sigConstrNames{\Sigma}
                \and f : \sigMonomialArity{\Sigma}(a) \to \Term_{\Sigma, V}
                \and c : \sigMonomialConst{\Sigma}(a)}
            {\constr{a}{f}{c} : \Term_{\Sigma, V}}
    \end{mathpar}
    Each $\var{x}$ represents a variable,
    and $\constr{a}{f}{c}$ represents the application of the operation named $a$ with arguments $f, c$.
\end{definition}

For the Plotkin powertheory of $D$, terms are constructed using inclusion and formal union operations.
Specifically, terms are constructed using the following derivation rules.
\begin{mathpar}
    \inferrule
        {d : D}
        {\powerIncl{d} : \Term_{\Sigma, V}}

    \inferrule
        {t_1, t_2 : \Term_{\Sigma, V}}
        {t_1 \formalUnion t_2 : \Term_{\Sigma, V}}
\end{mathpar}





\begin{definition}
  Let $\Sigma$ be a presignature.
  A \textbf{formal inequality} for $\Sigma$ consists of a type $V : \Universe$
  and two terms $t_1, t_2 : \Term_{\Sigma, V}$.
  We write $\Ineq_{\Sigma}$ for the type of formal inequalities for $\Sigma$.
\end{definition}

The type $V$ represents the variables in an inequality.
Typically, we write formal inequalities as $\formalIneqVar{t_1}{V}{t_2}$,
and we omit $V$ in case it is clear from the context.

In the Plotkin powertheory of $D$, we want the formal union to be commutative, associative and idempotent up to equality.
We can describe equalities using antisymmetry and two formal equalities.
However, it turns out that it is enough to only require the following formal inequalities.
\[
\formalIneq{x \formalUnion y}{y \formalUnion x} \quad \quad
\formalIneq{(x \formalUnion y) \formalUnion z}{x \formalUnion (y \formalUnion z)} \quad \quad
\formalIneq{x \formalUnion x}{x} \quad \quad
\formalIneq{x}{x \formalUnion x}
\]
We now construct these formal inequalities.
As all formal inequalities are constructed similarly, let us only look at the first one.
It contains two free variables, so we take $V = \Bool$.
We now define the terms representing the variables as $x = \var{\false}, y = \var{\true}$.
We then combine these using $\formalUnion \colon \Term_{\Sigma, V} \to \Term_{\Sigma, V} \to \Term_{\Sigma, V}$,
to create both the left- and right-hand side of the equation.

Finally, we put everything together and we define the notion of a \emph{signature}.
\begin{definition}
  A \textbf{signature} $\Sigma : \Sig$ consists of
  a presignature $\sigPre{\Sigma} : \PreSig$,
  a type of inequality names $\sigIneqNames{\Sigma} : \Universe$,
  and a family of formal inequalities $\sigIneqFam{\Sigma} : \sigIneqNames{\Sigma} \to \Ineq_{\sigPre{\Sigma}}$.
\end{definition}

All in all, a signature $\Sigma : \Sig$ describes operations and inequalities.
The operations are specifies by
a type of constructor names, $\sigConstrNames{\Sigma}$,
which specifies their names,
and a monomial $\sigMonomialFam{\Sigma}(a)$ for each $a : \sigConstrNames{\Sigma}$.
The arity and constant of an operation $a : \sigConstrNames{\Sigma}$ is given by a type $\sigMonomialArity{\Sigma}(a)$ and a DCPO $\sigMonomialConst{\Sigma}(a)$, respectively.
The inequalities are specified by a type $\sigIneqNames{\Sigma}$
and a family $\sigIneqFam{\Sigma}$ of formal inequalities over $\sigIneqNames{\Sigma}$.


We end this section by constructing a signature for the Plotkin powertheory of $D$.
We already defined its presignature.
We define $\sigIneqNames{\Sigma}$ as a four element type, as we want four formal inequalities.
The family of formal inequalities is defined using the formal inequalities we described above.

\begin{remark}
    Note that our terms and inequalities do not involve suprema contrasting them to presentations~(Definition~\ref{def:presentation}).
    However, when we define the initial algebra in Section~\ref{sec:initial-algebra}, we make use of the suprema in certain inequalities to require Scott continuity of the operations.
    Hence, an abstract basis~(Example~\ref{def:abstract-basis}) is insufficient, whereas a presentation is more general.
    The reason behind the choice of our definition of terms and inequalities is its similarity to what one would write down in algebra.
\end{remark}

\section{Algebras}
\label{sec:algebras}


Next we define algebras for the notion of signature that we defined in the previous section.
We do this in two steps.
First we look at \emph{prealgebras},
which consists of a DCPO together with the operations described by the signature.

\begin{definition}
  \label{def:prealgebra}
  Let $\Sigma$ be a signature.
  A \textbf{prealgebra} for $\Sigma$ consists of
  a DCPO $X$, called the underlying DCPO,
  and a Scott continuous map $\interpretation{X}{a} : \interp{\sigMonomialFam{\Sigma}(a)}(X) \to X$
  for each constructor name $a : \sigConstrNames{\Sigma}$.
  The type of prealgebras for $\Sigma$ is denoted by $\PreAlg{\Sigma}$.
\end{definition}

We can consider prealgebras as algebras for a functor.
Each presignature $\Sigma$ gives rise to a \emph{polynomial} functor $\interp{\Sigma}$,
which sends DCPOs $X$ to $\sum_{a : \sigConstrNames{\Sigma}} (\sigMonomialArity{\Sigma}(a) \to X) \times \sigMonomialConst{\Sigma}(a)$.
A prealgebra for $\Sigma$ is the same as an algebra for the functor $\interp{\Sigma}$.
The reason why we phrase it slightly differently in Definition \ref{def:prealgebra}, is because this way we do not need to mention type indexed coproducts of DCPOs.

Next we define morphisms between prealgebras.

\begin{definition}
  \label{def:prealgebra-mor}
  Let $X, Y$ be prealgebras for a signature $\Sigma$.
  An \textbf{algebra morphism} from $X$ to $Y$ is a Scott continuous map $f : X \to Y$ such that,
  for each $a : \sigConstrNames{\Sigma}$, we have $\interpretation{Y}{a} \circ \interp{\sigMonomialFam{\Sigma}(a)}(f) = f \circ \interpretation{X}{a}$.
\end{definition}

Every signature $\Sigma$ gives rise to a category whose objects are prealgebras and whose morphisms are algebra morphisms.
Next we define \emph{algebras}, which are prealgebras in which each formal inequality of the signature is satisfied.
To do so, we first show how to interpret the terms of formal inequalities.

\begin{definition}
  \label{def:terms-interpretation}
    Let $\Sigma$ be a presignature, $V$ a type of variables, $t : \Term_{\Sigma, V}$ a term,
    $X$ a prealgebra for $\Sigma$, and $\varAssignment : V \to X$ a variable assignment.
    We define $\interp{t}_{X, \varAssignment} : X$, the \textbf{interpretation} of $t$ in $X$,
    by recursion. 
    \begin{mathpar}
        \interp{\var{x}}_{X, \varAssignment} = \varAssignment(x)

        \interp{\constr{a}{f}{c}}_{X, \varAssignment} = \interpretation{X}{a}((\lambda b.\ \interp{f(b)}_{X, \varAssignment}), c)
    \end{mathpar}
    Variables get interpreted via $\varAssignment$.
    The term $\constr{a}{f}{c}$, representing the application of the operation named $a$ with arguments $f$ and $c$,
    is interpreted by performing the operation named $a$ in the prealgebra $X$.
    We often omit the subscript $X$, if it is clear from context which prealgebra we are working in.
\end{definition}

Note that this interpretation is natural in the prealgebra $X$.
\begin{proposition}
    Let $\Sigma$ be a presignature, $V$ a type of variables, and $t : \Term_{\Sigma, V}$ a term.
    Furthermore, let $X, Y$ be prealgebras for $\Sigma$ and $f : X \to Y$ a prealgebra morphism.
    Then, for each variable assignment $\varAssignment : V \to X$, we have that $\interp{t}_{Y, f \circ \varAssignment} = f(\interp{t}_{X, \varAssignment})$.
\end{proposition}


\begin{proof}
    We prove this by structural induction on $t$.
    In the case that $t$ is $\var{x}$, we conclude with reflexivity on $f(\varAssignment(x))$.
    Otherwise, if $t$ is $\constr{a}{g}{c}$, we have that
    \begin{align*}
        \interp{\constr{a}{g}{c}}_{Y, f \circ \varAssignment}
          &= \interpretation{Y}{a}((\lambda b.\ \interp{g(b)}_{Y, f \circ \varAssignment}), c) \\
          &= \interpretation{Y}{a}((\lambda b.\ f(\interp{g(b)}_{X, \varAssignment})), c) \\
          &= \interpretation{Y}{a}(\interp{\sigMonomialFam{\Sigma}(a)}(f)((\lambda b.\ \interp{g(b)}_{X, \varAssignment}), c)) \\
          &= f(\interpretation{X}{a}((\lambda b.\ \interp{g(b)}_{X, \varAssignment}), c)) \\
          &= f(\interp{\constr{a}{g}{c}}_{X, \varAssignment})
    \end{align*}
    We can use the induction hypothesis on $g(b) : \Term_{\Sigma, V}$,
    since all $g(b)$ are structurally smaller than $t$.
\end{proof}

With this interpretation at hand, we define what it means for a formal inequality to be valid in a prealgebra.
Intuitively, a formal inequality is valid, if the left-hand side is less than or equal to the right-hand side,
for all possible values for the variables. This is exactly captured by the following definition.

\begin{definition}
  \label{def:ineq-validity}
    Let $\formalIneqVar{t_1}{V}{t_2}$ be a formal inequality over $\Sigma$ with variables $V : \Universe$.
    Furthermore, let $X$ be a prealgebra for $\Sigma$.
    We say that this formal inequality is \textbf{valid} in $X$ if,
    for each $\varAssignment : V \to X$,
    \begin{gather*}
        \interp{t_1}_\varAssignment \sqsubseteq \interp{t_2}_\varAssignment
    \end{gather*}
    Note that this is an inequality in the underlying DCPO of $X$.
    Whenever such a formal inequality is valid in $X$,
    we write it using the usual notation for validity: $X \vDash \formalIneqVar{t_1}{V}{t_2}$.
\end{definition}

An algebra for a signature can now be defined as a prealgebra for that signature,
where all the formal inequalities are valid.

\begin{definition}
    Let $\Sigma$ be a signature. An \textbf{algebra} for $\Sigma$ is a prealgebra for $\Sigma$,
    such that for each $j : \sigIneqNames{\Sigma}$, we have that $X \vDash \sigIneqFam{\Sigma}(j)$.
    We define $\Alg{\Sigma}$ to be the full subcategory of $\PreAlg{\Sigma}$,
    of those prealgebras where all $\sigIneqFam{\Sigma}(j)$ are valid.
\end{definition}

\section{Initial Algebra}
\label{sec:initial-algebra}
We now construct the initial algebra for a signature as a {\em
  quotient inductive-inductive type}, QIIT. In
Section~\ref{sec:prelims-qiits} we have already motivated QIITs and in
Section~\ref{sec:QIIT} we illustrated the use of QIITs as initial DCPO
algebras using the example of the powerdomain in Figures
\ref{fig:power-domain-dcpo}, \ref{fig:power-domain-operations} and \ref{fig:power-domain-ineqs}.

Our goal in this section is to adapt this construction to cover arbitrary signatures.

Recall that we subdivided the constructors of this QIIT into three groups.
The first group, given in Figure~\ref{fig:power-domain-dcpo},
expresses that the type being constructed is a DCPO.
More specifically, it says that we have a partial order
and a supremum operation.
We do not need to modify any of these constructors,
because again we are constructing a DCPO.
What we do need to change, however, are the constructors in Figures \ref{fig:power-domain-operations} and \ref{fig:power-domain-ineqs}.
In these two figures, the operations and inequalities of the signatures are described respectively.
Specifically,
we need to change these constructors so that instead we have the operations and inequalities described by the signature.
Since QIITs give us initial algebras~\cite{DBLP:conf/fossacs/AltenkirchCDKF18},
it follows that the constructed DCPO algebra is indeed initial,
which we prove in this section.


We start this section by constructing the initial algebra for a signature in Section~\ref{sec:initial-algebra-qiit}.
Then we show its initiality in Section~\ref{sec:initial-algebra-initiality} after showing the recursion and induction principles of the QIIT.

\subsection{Construction of the Initial Algebra}
\label{sec:initial-algebra-qiit}

To construct the initial algebra for a signature $\Sigma$,
we simultaneously define a type $\Initial{\Sigma}$ and an ordering relation $\Leq{\Sigma}$ on $\Initial{\Sigma}$ as a QIIT.
We discuss their constructors in three steps.
First, we show that we have constructors to define a DCPO structure on $\Initial{\Sigma}$.
We then show that we have the constructors to lift this DCPO structure to a prealgebra structure for $\Sigma$.
Finally, we lift this structure to an algebra structure for $\Sigma$.

\begin{definition}
    We simultaneously define the type $\Initial{\Sigma}$,
    and its ordering relation $\Leq{\Sigma}$ as a QIIT.
    Their constructors are given in Figures~\ref{fig:initial-dcpo-constructors},\ref{fig:initial-prealg-constructors},\ref{fig:initial-alg-constructors}.
    The constructors of $\Initial{\Sigma}$ are presented as terms with a type,
    while the constructors of $\Leq{\Sigma}$ are given as inference rules.
    We quantify over directed families whose index types live in some universe $\Universe$,
    so both $\Initial{\Sigma}$ and the truth values of $\Leq{\Sigma}$ live in a larger universe than $\Universe$,
    and, if the arities and constants of the monomials in $\Sigma$ live in $\Universe^+$,
    both $\Initial{\Sigma}$ and the truth values of $\Leq{\Sigma}$ also live in $\Universe^+$.
\end{definition}

\begin{figure}
    \begin{mathpar}
        \inferrule[\textsc{Leq-Refl}]
            { }
            {x \Leq{\Sigma} x}
            [\textsc{Leq-Refl}]\label{rule:leq-refl}

        \inferrule[\textsc{Leq-Trans}]
            {x \Leq{\Sigma} y \and y \Leq{\Sigma} z}
            {x \Leq{\Sigma} z}
            [\textsc{Leq-Trans}]\label{rule:leq-trans}

        \inferrule[\textsc{Leq-Prop}]
            { }
            {\isProp(x \Leq{\Sigma} y)}
            [\textsc{Leq-Prop}]\label{rule:leq-prop}
        \\
        \inferrule[\textsc{Leq-Antisym}]
            {x \Leq{\Sigma} y \and y \Leq{\Sigma}}
            {x = y}
            [\textsc{Leq-Antisym}]\label{rule:leq-antisym}

        \inferrule[\textsc{Initial-Set}]
            { }
            {\isSet(\Initial{\Sigma})}
            [\textsc{Initial-Set}]\label{rule:initial-set}
        \\
        \initialLub{\Sigma} : \prod_{\alpha : I \to \Initial{\Sigma}} \isDirected{\alpha} \to \Initial{\Sigma}
        \\
        \inferrule[\textsc{$\initialLub{\Sigma}$-Is-Sup-1}]
            {\alpha : I \to \Initial{\Sigma} \and \delta : \isDirected{\alpha}}
            {\prod_{i : I} \alpha(i) \Leq{\Sigma} \initialLub{\Sigma}(\alpha, \delta)}
            [\textsc{$\initialLub{\Sigma}$-Is-Sup-1}]\label{rule:leq-lub-is-upperbound}

        \inferrule[\textsc{$\initialLub{\Sigma}$-Is-Sup-2}]
            {\alpha : I \to \Initial{\Sigma} \and \delta : \isDirected{\alpha}}
            {\prod_{v : \Initial{\Sigma}} \isUpperbound(v, \alpha) \to \initialLub{\Sigma}(\alpha, \delta) \Leq{\Sigma} v}
            [\textsc{$\initialLub{\Sigma}$-Is-Sup-2}]\label{rule:leq-lub-is-lowerbound-of-upperbounds}
    \end{mathpar}
    \caption{Constructors for the DCPO structure}
    \label{fig:initial-dcpo-constructors}
    \vspace{\floatsep}

    \begin{mathpar}
        \app{a} : \interp{\sigMonomialFam{\Sigma}(a)}(\Initial{\Sigma}) \to \Initial{\Sigma}

        \inferrule[\textsc{$\app{}$-Is-Cont-1}]
            {a : \sigConstrNames{\Sigma} \and \alpha : I \to \interp{\sigMonomialFam{\Sigma}(a)}(\Initial{\Sigma})
                \and \delta : \isDirected{\alpha}}
            {\prod_{i : I} \app{a}(\alpha(i)) \Leq{\Sigma} \app{a}(\dcpoSup{\interp{\sigMonomialFam{\Sigma}(a)}(\Initial{\Sigma})} \alpha)}
            [\textsc{$\app{}$-Is-Cont-1}]\label{rule:leq-app-cont-upperbound}

        \inferrule[\textsc{$\app{}$-Is-Cont-2}]
            {a : \sigConstrNames{\Sigma} \and \alpha : I \to \interp{\sigMonomialFam{\Sigma}(a)}(\Initial{\Sigma})
                \and \delta : \isDirected{\alpha}}
            {\prod_{v : \Initial{\Sigma}} \isUpperbound(v, {\app{a}} \circ \alpha) \to \app{a}(\dcpoSup{\interp{\sigMonomialFam{\Sigma}(a)}(\Initial{\Sigma})} \alpha) \Leq{\Sigma} v}
            [\textsc{$\app{}$-Is-Cont-2}]\label{rule:leq-app-cont-lowerbound-of-upperbounds}
    \end{mathpar}
    \caption{Constructors for the prealgebra structure for $\Sigma$}
    \label{fig:initial-prealg-constructors}
    \vspace{\floatsep}

    \begin{mathpar}
        \inferrule[\textsc{Ineq-Valid}]
            {j : \sigIneqNames{\Sigma} \and \sigIneqFam{\Sigma}(j) = \formalIneqVar{t_1}{V}{t_2} \and \rho : V \to \Initial{\Sigma}}
            {\interp{t_1}_\varAssignment \Leq{\Sigma} \interp{t_2}_\varAssignment}
            [\textsc{Ineq-Valid}]\label{rule:leq-ineq-valid}
    \end{mathpar}
    \caption{Constructors for the algebra structure for $\Sigma$}
    \label{fig:initial-alg-constructors}
\end{figure}

Compared to the powerdomain example from Section~\ref{sec:QIIT}, the QIIT here is presented differently.
In particular, the two constructors for the inclusion and union operations are replaced by $\app{a}$ constructors for each constructor name $a : \sigConstrNames{\Sigma}$.
Moreover, the continuity of each $\app{a}$ is written down by expressing that $\app{a}$ sends directed suprema to suprema in $\Initial{\Sigma}$.
This notion of continuity~\cite[Definition~13]{dejong:2021}, means we do not need to add path constructors to the QIIT to express continuity,
like we did in Section~\ref{sec:QIIT}.

\begin{remark}\label{remark:comparison-w-types-const-dcpo}
    If we unfold the definition of the interpretation of a monomial, we see that $\app{}$ has the following type:
    $\prod_{a : \sigConstrNames{\Sigma}} (\sigMonomialArity{\Sigma}(a) \to \Initial{\Sigma}) \times \sigMonomialConst{\Sigma}(a) \to \Initial{\Sigma}$.
    The introduction rule for W-types~\cite{Lof84} is similar, but does not include $\sigMonomialConst{\Sigma}(a)$.
    Like we noted in Remark~\ref{remark:signature-w-types-comparison},
    we cannot leave out $\sigMonomialConst{\Sigma}(a)$, as this would result in a loss of structure.
    This is also apparent here, as otherwise we would not be able to state that $\app{a}$ is Scott continuous with respect to the constant DCPO $\sigMonomialConst{\Sigma}(a)$.
\end{remark}

The constructors for the DCPO structure on $\Initial{\Sigma}$ are given in Figure~\ref{fig:initial-dcpo-constructors}.
The rules \nameref{rule:leq-refl}, \nameref{rule:leq-trans} and \nameref{rule:leq-prop}
make sure that $\Leq{\Sigma}$ is reflexive, transitive and proposition valued.
Furthermore, we guarantee that $\Leq{\Sigma}$ is antisymmetric
and $\Initial{\Sigma}$ is a set,%
\footnote{We could do without the set truncation,
    as the underlying type of a partial order is always a set~\cite[Definition 4.1]{JongEscardo23}.
    See also \url{https://github.com/martinescardo/TypeTopology/blob/02add316f8a78bc79e5687e4249d9f0174dae1d2/source/UF/Hedberg.lagda\#L90}.}
by adding suitable path constructors.
Finally, the $\initialLub{\Sigma}$ constructor is supposed to give the supremum for directed families.
This is guaranteed by \nameref{rule:leq-lub-is-upperbound} and \nameref{rule:leq-lub-is-lowerbound-of-upperbounds}.
Note that we often write $\initialLub{\Sigma} \alpha$, in case we do not care about the specific proof of directedness.

\begin{lemma}\label{lemma:initial-dcpo-structure}
    The type $\Initial{\Sigma}$ has a DCPO structure.
\end{lemma}

Next up, the constructors for the prealgebra structure are given in Figure~\ref{fig:initial-prealg-constructors}.
The $\app{a}$ constructor introduces, for each $a : \sigConstrNames{\Sigma}$,
an operation corresponding to the monomial $\sigMonomialFam{\Sigma}(a)$.
Note that, by Lemma~\ref{lemma:initial-dcpo-structure}, we can actually pass $\Initial{\Sigma}$
as an argument to the DCPO endo-functor $\interp{\sigMonomialFam{\Sigma}(a)}$.
To guarantee that each operation $\app{a}$ is continuous, we have the rules
\nameref{rule:leq-app-cont-upperbound} and \nameref{rule:leq-app-cont-lowerbound-of-upperbounds}.
Combining these rules we get
\begin{gather*}
    \inferrule
        {a : \sigConstrNames{\Sigma} \and \alpha : I \to \interp{\sigMonomialFam{\Sigma}(a)}(\Initial{\Sigma})
            \and \delta : \isDirected{\alpha}}
        {\isLeastUpperbound(\app{a}(\dcpoSup{\interp{\sigMonomialFam{\Sigma}(a)}(\Initial{\Sigma})} \alpha), {\app{a}} \circ \alpha)}
\end{gather*}
This is exactly what it means for $\app{a}$ to be continuous.

\begin{lemma}\label{lemma:initial-prealg-structure}
    The type $\Initial{\Sigma}$ has a prealgebra structure for the signature $\Sigma$.
\end{lemma}

Finally, we have one single constructor for the algebra structure for the signature $\Sigma$ given in Figure~\ref{fig:initial-prealg-constructors}.
The rule \nameref{rule:leq-ineq-valid} guarantees that all the formal inequalities of the signature are valid.
Note that, by Lemma~\ref{lemma:initial-prealg-structure}, we have an algebra structure for $\sigPre{\Sigma}$,
so we can indeed interpret the left- and right-hand side of the formal inequality in $\Initial{\Sigma}$.

\begin{theorem}\label{theorem:initial-prealg-structure}
    The type $\Initial{\Sigma}$ has an algebra structure for the signature $\Sigma$.
\end{theorem}

\subsection{Initiality}\label{sec:initial-algebra-initiality}

To show that $\Initial{\Sigma}$ is initial in $\Alg{\Sigma}$,
we need to construct a unique morphism $\Initial{\Sigma} \to X$ for arbitrary $X$.
The existence of such a morphism comes from the recursion principle for $\Initial{\Sigma}$.

\begin{theorem}
    Let $\Sigma$ be a signature and $X$ an algebra for $\Sigma$.
    Then there exists a morphism ${\rec_{\Sigma, X} : \Initial{\Sigma} \to X}$,
    with the following computation rules%
    \footnote{Here, $\doteq$ means definitional equality.}
    \begin{align*}
        \rec_{\Sigma, X}(\app{a}(x)) &\doteq \interpretation{X}{a}(\interp{\sigMonomialFam{\Sigma}(a)}(\rec_{\Sigma, X})(x)) \\
        \rec_{\Sigma, X}(\initialLub{\Sigma} \alpha) &\doteq \dcpoSup{X} (\rec_{\Sigma,X} \circ \alpha)
    \end{align*}
\end{theorem}
\begin{proof}
    The two remaining constructors of $\Initial{\Sigma}$, \nameref{rule:leq-antisym} and \nameref{rule:initial-set},
    get send to the proofs that $X$ has an antisymmetric relation and is a set, respectively.
    The function $\rec_{\Sigma,X}$ is well-defined, as it is monotone.%
    \footnote{In Agda, we annotate the recursion principle with a \texttt{TERMINATING} pragma,
        as the termination checker believes there are non-structural recursive calls in the monotonicity proof.
        We think that the termination checker is too restrictive here,
        and that the recursion principle should be accepted.
        The problem seems to be that Agda does not unfold definitions enough to recognize that the recursive calls are correct.
        We explain this in more detail in the formalization (\texttt{src/SignatureAlgebra/Initial/Elimination.agda}).}
    It is an algebra morphism, as by definition, it is continuous and commutes with all the operations.
\end{proof}

To prove the uniqueness of such a morphism, we need an induction principle for $\Initial{\Sigma}$.
For this, we could use \emph{displayed algebras}~\cite{DBLP:conf/popl/Sojakova15,DBLP:conf/rta/KaposiK18}
in its general form, but rather, in this paper, we only consider our specific use case.
First, let us define how a monomial acts on a family of types.
\begin{definition}
    Let $M$ be a monomial with arity $B$ and constant DCPO $C$,
    $X$ a DCPO, and $Y : X \to \Universe$ a family of types.
    The family $Y$ lifts to a type family $\monomialAction{M}(Y) : \interp{M}(X) \to \Universe$, which is defined as
    \begin{gather*}
        \monomialAction{M}(Y)(f, c) = \left(\prod_{b : B} Y(f(b))\right) \times C
    \end{gather*}
\end{definition}
The action defined above, allows us to state the induction hypotheses for $\Initial{\Sigma}$.
If we restrict our view to a predicate $Y : X \to \Universe$ and let $x : \interp{M}(X)$,
then $\monomialAction{M}(Y)(x)$ expresses the fact that $Y$ holds for all the subterms of type $X$ in $x$.
This is exactly what we need as the induction hypothesis in the case of $\app{a}(x) : \Initial{\Sigma}$.
\begin{theorem}\label{thm:initial-alg-ind}
    Let $Y : \Initial{\Sigma} \to \Universe$ be a family of types, such that
    \begin{itemize}
        \item each $Y(x)$ is a proposition;
        \item $Y$ is \emph{supremum preserving}: if $\alpha : I \to \Initial{\Sigma}$
              is a directed family, such that $Y(\alpha(i))$ holds for every $i : I$, then $Y(\initialLub{\Sigma}(\alpha))$;
        \item $Y$ is \emph{app preserving}: if $a$ is a constructor name,
              $x : \interp{\sigMonomialFam{\Sigma}(a)}(\Initial{\Sigma})$ and $\monomialAction{\sigMonomialFam{\Sigma}(a)}(Y)(x)$,
              then $Y(\app{a}(x))$ holds.
    \end{itemize}
    In that case $Y(x)$ holds for every $x : \Initial{\Sigma}$, that is,
    we have a function $\ind_{\Sigma,Y} : \prod_{x : \Initial{\Sigma}} Y(x)$.
\end{theorem}
\begin{proof}
    The function $\ind_{\Sigma,Y}$ is defined by pattern matching.
    For the $\initialLub{\Sigma}$ and $\app{a}$ constructors, we use that fact that $Y$ is supremum and app preserving.
    On the path constructors of $\Initial{\Sigma}$, $\ind_{\Sigma,Y}$ is defined using the fact that $Y$ is a proposition valued.
\end{proof}

With these elimination principles at hand, we show the initiality of $\Initial{\Sigma}$.
\begin{theorem}
  \label{thm:initial-alg}
    Let $\Sigma$ be a signature.
    The algebra $\Initial{\Sigma}$ is initial in the category $\Alg{\Sigma}$.
\end{theorem}
\begin{proof}
    Let $X$ be an arbitrary algebra for $\Sigma$.
    The unique algebra morphism $\bang : \Initial{\Sigma} \to X$ is given by the recursion principle.

    Now let $f : \Initial{\Sigma} \to X$ be an arbitrary map of algebras.
    To show that $\bang$ and $f$ are equal, it is enough to show that their underlying functions are equal.
    After applying function extensionality, it thus suffices to show the following for each $x : \Initial{\Sigma}$
    \begin{gather}\label{eqn:initiality-ind-statement}
        \bang(x) = f(x)
    \end{gather}
    As the underlying type $X$ is a set, we know that \eqref{eqn:initiality-ind-statement} is a proposition.
    We can thus use the induction principle of $\Initial{\Sigma}$.
    First we prove that \eqref{eqn:initiality-ind-statement} is true for the supremum of a directed family.
    Let $\alpha : I \to \Initial{\Sigma}$ be a directed family.
    By the computation rules of $\bang$ and the continuity of $f$,
    we need to show that the suprema in $X$ of $\bang \circ \alpha$ and $f \circ \alpha$ are equal.
    This holds, because the families $\bang \circ \alpha$ and $f \circ \alpha$ are equal by the induction hypothesis.

    Finally, we need to show that \eqref{eqn:initiality-ind-statement} is true for any $\app{a}(x)$.
    By the computation rules of $\bang$ and the fact that $f$ commutes with all operations,
    we have to show that
    \begin{gather*}
        \interpretation{X}{a}(\interp{\sigMonomialFam{\Sigma}(a)}(\bang)(x)) = \interpretation{X}{a}(\interp{\sigMonomialFam{\Sigma}(a)}(f)(x))
    \end{gather*}
    This follows from the induction hypothesis, as \eqref{eqn:initiality-ind-statement} holds for all the subterms of $x$.
\end{proof}

\section{Examples}\label{sec:examples}
Finally, we show various examples of signatures and their initial algebras in this section
using Theorem \ref{thm:initial-alg}.
We start with several categorical constructions of DCPOs (coalesced sums, smash products, and coequalizers),
and then we show examples of free DCPOs (partiality and powerdomains).

\subsection{Coalesced Sum}
Recall that a pointed DCPO is a DCPO with a least element.
The categorical coproduct of pointed DCPOs is defined as their coalesced sum.
Classically, one defines this by taking the union of the two pointed DCPOs,
removing both their bottom elements, and adding a new bottom element.
We define this DCPO as the initial algebra for a signature.

\begin{example}
  Let $D, E$ be pointed DCPOs. We define $\csumSig{D}{E}$.
  We add two operations, with constructor names $\sumInclName{D}, \sumInclName{E}$
  and corresponding monomials $(\bot \to X) \times D$ and $(\bot \to X) \times E$.
  This allows us to write down terms of the form $\sumIncl{D}(d), \sumIncl{E}(e)$ for $d : D, e : E$.
  Finally, we add the following two inequalities.
  \[
  \formalIneq{\sumIncl{D}(\bot_D)}{\sumIncl{E}(\bot_E)} \quad \quad
  \formalIneq{\sumIncl{E}(\bot_E)}{\sumIncl{D}(\bot_D)}
  \]
\end{example}

An algebra for $\csumSig{D}{E}$ consists of a DCPO $X$,
together with inclusion operations $\sumIncl{D} : D \to X$ and $\sumIncl{E} : E \to X$,
such that $\sumIncl{D}(\bot_D) = \sumIncl{E}(\bot_E)$.

\begin{proposition}
    Let $D, E$ be pointed DCPOs. Their coproduct is given by the initial algebra for $\csumSig{D}{E}$.
\end{proposition}
\begin{proof}
    By taking $\sumIncl{D}(\bot_D) = \sumIncl{E}(\bot_E)$ as the bottom element and using Theorem~\ref{thm:initial-alg-ind}.
\end{proof}

\subsection{Coequalizer}
Let $D, E$ be DCPOs and $f, g : D \to E$ continuous maps.
We define the coequalizer by first defining a signature whose algebras correspond to coequalizer cocones,
and then considering the initial algebra for this signature.

\begin{example}
    Let $D,E$ be DCPOs and $f, g : D \to E$. We define $\coeqSig{f}{g}$.
    It has a single operation which includes $E$, i.e. it is represented by the monomial $(\bot \to X) \times E$.
    This allows us to write terms of the form $\coeqIncl(e)$ for $e : E$.
    For each $d : D$, we add the following two inequalities.
    \[
      \formalIneq{\coeqIncl(f(d))}{\coeqIncl(g(d))} \quad \quad
      \formalIneq{\coeqIncl(g(d))}{\coeqIncl(f(d))}
    \]
\end{example}

An algebra for $\coeqSig{f}{g}$ is a DCPO $X$,
together with a map $\coeqIncl : E \to X$,
such that $\coeqIncl(f(d)) = \coeqIncl(g(d))$ for each $d : D$.

\begin{proposition}
  Let $D,E$ be DCPOs and $f, g : D \to E$.
  Algebras for $\coeqSig{f}{g}$ correspond to coequalizer cocones of $f$ and $g$.
  Furthermore, the initial algebra for $\coeqSig{f}{g}$ is the coequalizer of $f$ and $g$.
\end{proposition}

\subsection{Smash Product}
Let $D,E$ be pointed DCPOs. Their smash product is defined to be the product $D \times E$,
where we identify $(\bot_D, e) = (d, \bot_E)$ for every $d : D, e : E$.
Concretely, we define the smash product as the initial algebra for a suitably chosen signature.

\begin{example}
  Let $D,E$ be pointed DCPOs. We define $\smashSig{D}{E}$.
  It has a single operation which includes $D \times E$, i.e. it is represented by the monomial $(\bot \to X) \times (D \times E)$.
  This allows us to write terms of the form $(d, e)$ for $d : D, e : E$.
  For each $d : D, e : E$, we add the following two inequalities.
  \[
    \formalIneq{(d, \bot_E)}{(\bot_D, e)} \quad \quad
    \formalIneq{(\bot_D, e)}{(d, \bot_E)}
  \]
\end{example}

An algebra for $\smashSig{D}{E}$ consists of a DCPO $X$,
together with an inclusion $\smashIncl : D \times E \to X$,
such that $\smashIncl(\bot_D, e) = \smashIncl(d, \bot_E)$ for every $d : D, e : E$.

\begin{definition}
    Let $D, E$ be pointed DCPOs. Their \textbf{smash product} is the initial algebra for $\smashSig{D}{E}$.
\end{definition}

\subsection{Free Algebra for a Signature}
\label{sec:initial-algebra-free}
The next two examples that we consider are special cases of a more general construction,
namely the free algebra,
which we discuss here.
More specifically, suppose that we have some signature $\Sigma$ and some DCPO $D$.
Our goal is to define a signature $\sigOver{\Sigma}{D}$ such that the initial algebra for $\sigOver{\Sigma}{D}$ gives the free $\Sigma$-algebra for $D$.
We use this construction to construct a left adjoint $F : \DCPO \to \Alg{\Sigma}$ for the forgetful functor $U : \Alg{\Sigma} \to \DCPO$.
The main idea is that $\sigOver{\Sigma}{D}$ is obtained by adding an operation $D \to UFD$ to $\Sigma$.

\begin{definition}
  Let $\Sigma$ be a signature and $D$ a DCPO.
  We define a signature $\sigOver{\Sigma}{D}$, by extending the constructor names of $\Sigma$ with a new element $\inclName$,
  and assign $(\bot \to X) \times D$ as the monomial for $\inclName$.
  We do not add any new inequalities, but note that we have to lift the original inequalities of $\Sigma$ as we changed the type of constructor names.
  We leave these technical details to the formalization.
\end{definition}

If $X : \Alg{\Sigma}$ is an algebra and $f : D \to U X$,
we construct an algebra for the signature $\sigOver{\Sigma}{D}$,
by using $f$ as the interpretation for the inclusion constructor name $\inclName$.
We write $\addIncl{X}{f} : \Alg{\sigOver{\Sigma}{D}}$ for this algebra.
Conversely, if $X : \Alg{\sigOver{\Sigma}{D}}$ is an algebra for $\sigOver{\Sigma}{D}$, we define the algebra $\forgetIncl{X} : \Alg{\Sigma}$,
by forgetting the interpretation of $\inclName$.
Now we define the free $\Sigma$-algebra for $D$ to be $\forgetIncl{(\Initial{\sigOver{\Sigma}{D}})}$.

\begin{theorem}
  \label{thm:free-alg}
  The operation $D \mapsto \forgetIncl{(\Initial{\sigOver{\Sigma}{D}})}$ lifts to a functor $F : \DCPO \to \Alg{\Sigma}$,
  which is left adjoint to the forgetful functor $U : \Alg{\Sigma} \to \DCPO$.
\end{theorem}
\begin{proof}
    We construct this adjunction using universal arrows.
    The unit of the adjunction, $\eta_D : D \to UFD$, is given by the inclusion operation $\app{\inclName}$.
    The universal arrows are constructed in the following way.
    Let $X : \Alg{\Sigma}$ be an algebra and $f : D \to UX$.
    By the initiality of $\Initial{\sigOver{\Sigma}{D}}$, we have a unique map $\Initial{\sigOver{\Sigma}{D}} \to \addIncl{X}{f}$.
    By forgetting about the fact that this map commutes with the inclusion operation,
    we get a map $FD \to X$, whose uniqueness follows from the initiality of $\Initial{\sigOver{\Sigma}{D}}$.
\end{proof}

\subsection{Pointed DCPO}
We instantiate the free algebra construction of Section~\ref{sec:initial-algebra-free} to two specific examples.
The first example that we consider, is given by pointed DCPOs.
The type of pointed DCPOs can be given as algebras for a suitably chosen signature.
\begin{example}\label{ex:pointed-dcpo-sig}
  We define $\pointedDcpoSig$, the signature for pointed DCPOs.
  We add one operation with constructor name $\botName$ and monomial $(\bot \to X) \times 1$.
  This allows us to write down the term $\botTerm$, which represents the least element.
  Finally, we add the formal inequality $\formalIneq{\botTerm}{x}$ to our signature.
\end{example}

An algebra for $\pointedDcpoSig$ thus consists of a DCPO $D$ together with a least element.
A map between algebras $D,E$ is a continuous map $D \to E$ which preserves the bottom element.

\begin{proposition}\label{prop:pointed-dcpo-as-alg}
    Algebras for $\pointedDcpoSig$ correspond to pointed DCPOs.
\end{proposition}

By using Theorem \ref{thm:free-alg} for the signature $\pointedDcpoSig$,
we obtain an adjunction between the category of DCPOs and pointed DCPOs,
and this adjunction gives rise to a monad.
Note that this construction is similar to how Altenkirch, Danielsson, and Kraus constructed the partiality monad~\cite{Partiality},
except that we use DCPOs instead of $\omega$-CPOs.

\subsection{Power Algebras}
Our final example is the Plotkin powerdomain for which we described a suitable signature in Section \ref{sec:signatures}.

\begin{example}\label{ex:power-sig}
  We define $\powerSig$, the signature for the Plotkin powertheory.
  We add one operation with constructor name $\formalUnionName$ and monomial $(\Bool \to X) \times 1$.
  Given terms $t_1, t_2$, we write $t_1 \formalUnion t_2$ for the term representing the formal union of $t_1, t_2$.
  Finally, we add the following formal inequalities, to guarantee that formal union is commutative, associative and idempotent.
  \[
    \formalIneq{x \formalUnion y}{y \formalUnion x} \quad \quad
    \formalIneq{(x \formalUnion y) \formalUnion z}{x \formalUnion (y \formalUnion z)} \quad \quad
    \formalIneq{x \formalUnion x}{x} \quad \quad
    \formalIneq{x}{x \formalUnion x}
  \]
\end{example}

An algebra for $\powerSig$ consists of a DCPO $D$, and a continuous map $\formalUnion : D \to D \to D$
which is commutative, associative and idempotent.
A map between algebras $D,E$ is a continuous map $D \to E$ which commutes with the formal union operations.

By using Theorem \ref{thm:free-alg} for the signature $\powerSig$,
we get that every DCPO $D$ gives rise to a free algebra for $\powerSig$,
which is the Plotkin powerdomain.
This free algebra is constructed as the QIIT with the constructors given in Figure \ref{fig:power-domain-operations}, \ref{fig:power-domain-ineqs}, and \ref{fig:power-domain-dcpo}.

\section{Related Work}
\label{sec:related-work}
In this section, we briefly recall the related work and how it compares to the results in this paper.
Our development is inspired by the work of Altenkirch, Danielsson, and Kraus~\cite{Partiality} that discusses the partiality monad as a domain. We extend this by developing a general notion of signature for DCPOs to capture a large variety of domain constructions.
A difference is that Altenkirch, Danielsson, and Kraus use $\omega$-CPOs, whereas we use DCPOs.
These two notions are equivalent if we assume the axiom of choice (and limit to countable index sets),
but constructively the notion of DCPO is stronger.
This difference complicates the specification of the QIIT in Section \ref{sec:initial-algebra},
because we have to quantify over all directed sets (and not just chains).
Note that one can modify our work to obtain free $\omega$-CPO algebras.

There are various other constructions of the partiality monad.
We have already mentioned the work by Escard\'o and Knapp~\cite{escardo:2017} and by De Jong and Escard\'o~\cite{dejong:2021},
who construct the lift of types and of DCPOs without using any quotient type.
More specifically, the lift of a type $A$ is defined to be the type $\sum_{P : \Omega} P \rightarrow A$ of partial elements of $A$,
where $\Omega$ is the type of all propositions.
Their construction gives rise to a DCPO, which is in fact the free pointed DCPO.
Note that each of these constructions is isomorphic to each other, because they give rise to a left adjoint of the forgetful functor from the category of DCPOs to the category of sets.
Chapman, Uustalu, and Veltri construct the free $\omega$-CPO using a quotient inductive type, namely the free countably-complete semilattice~\cite{DBLP:conf/ictac/ChapmanUV15}.
Their definition is similar to the one used in the work by Escard\'o and Knapp and by De Jong and Escard\'o,
except that $\Omega$ is replaced by the initial $\sigma$-frame.
They use quotient inductive types to construct the initial $\sigma$-frame.
Note that one can also construct the initial $\sigma$-frame using only function extensionality, propositional extensionality, and propositional truncations
as shown by Escard\'o\footnote{\url{https://github.com/martinescardo/TypeTopology/blob/master/source/NotionsOfDecidability/QuasiDecidable.lagda}}.

Finally, there are various other constructions of free DCPO algebras in classical foundations.
One way to construct these algebras is by using the free DCPO for a presentation~\cite{JUNG2008209}.
This construction is impredicative because it uses the power set,
and if one were to use this in predicative foundations,
then it raises the universe level.
Bidlingmaier, Faissole, and Spitters discuss three ways to justify the existence of free $\omega$-CPOs for presentations constructively~\cite[Section 3]{bidlingmaier:2021},
and they assume the existence of free $\omega$-CPOs for presentations.
One can justify their axiom by either assuming impredicativity, countable choice, or QIITs,
and in this paper, we use the last way to justify the existence of free complete partial orders.
Abramsky and Jung~\cite[Chapter 6]{abramskyJung} discuss equational theories of domains.
Their notion of signature only includes operations with an arity,
and they do not have constants.
For this reason,
one cannot construct the coequalizers, the coalesced sum, and the smash product as an initial algebra of their signatures.

\section{Conclusions and Future Work}
\label{sec:conclusion}
We have presented a general framework to describe and construct
algebraic effects in domain theory.  We describe the algebraic effects
using a signature: a set of operations accompanied by an inequational
theory these operations should obey.  An algebra for such a signature
is a DCPO with Scott continuous maps for each operation, such that the
inequational theory is satisfied.  We construct the initial algebra as
a quotient inductive-inductive type (QIIT), extending~\cite{Partiality} (where the
partiality monad is constructed as a QIIT). 
Different from~\cite{JUNG2008209}, this construction does not use power sets, so it is predicative.
The initial algebra is used for interpreting algebraic effects in denotational
semantics and we show that our framework captures a variety of well-known
examples, like coalesced sums, smash products, coequalizers,
partiality and power domains.
Our work has all been formalized in Cubical Agda.

Potential future work is to also consider handlers for algebraic effects.
We focus on the effects themselves, but their handlers are just as important,
as a handler describes how a result of a particular effect should be handled~\cite{lmcs:705}.
To give the full denotational semantics of a language with algebraic effects,
we also need to interpret the handlers in domain theory.

\end{document}

%% file: setup.tex
\usepackage{stmaryrd}
\usepackage{mathtools}
\usepackage{mathpartir}

\makeatletter

\DeclareDocumentCommand \inferrule { s O {} m m o }{%
  \IfBooleanTF{#1}%
  {%
    \mpr@inferstar[#2]{#3}{#4}%
  }{%
    \mpr@inferrule[#2]{#3}{#4}%
  }%
  \IfValueT{#5}%
  {%
    \my@name@inferrule{#5}%
  }%
}
\NewDocumentCommand \my@name@inferrule { m }{%
  \def\@currentlabelname{\ensuremath{#1}}%
}
\makeatother

%% file: macros.tex
\newcommand{\interp}[1]{\llbracket #1 \rrbracket}

\DeclareMathOperator{\isProp}{isProp}
\DeclareMathOperator{\isSet}{isSet}

\renewcommand{\tt}{*}
\newcommand{\Bool}{\mathrm{Bool}}
\newcommand{\true}{\mathrm{true}}
\newcommand{\false}{\mathrm{false}}

\newcommand{\DCPO}{\mathrm{DCPO}}

\newcommand{\isDirected}[1]{\operatorname{isDirected}(#1)}
\DeclareMathOperator{\isUpperbound}{isUpperbound}

\DeclareMathOperator{\isLeastUpperbound}{isLeastUpperbound}
\newcommand{\dcpoSup}[1]{\textstyle\bigsqcup_{#1}}

\newcommand{\RoundedIdeal}[1]{\textrm{R-Idl}(#1)}
\newcommand{\covered}{\triangleleft}
\newcommand{\Ideal}[1]{\mathrm{Idl(#1)}}

\newcommand{\Monomial}{\mathrm{Monomial}}
\newcommand{\monomial}{M}
\newcommand{\monomialArity}[1]{{#1}_B}
\newcommand{\monomialConst}[1]{{#1}_C}
\newcommand{\monomialAction}[1]{\overline{#1}}

\newcommand{\Sig}{\mathrm{Sig}}
\newcommand{\PreSig}{\mathrm{PreSig}}
\newcommand{\sigConstrNames}[1]{{#1}_A}
\newcommand{\sigMonomialFam}[1]{{#1}_{\monomial}}
\newcommand{\sigMonomialArity}[1]{{#1}_B}
\newcommand{\sigMonomialConst}[1]{{#1}_C}
\newcommand{\sigPre}[1]{{#1}_{\mathrm{pre}}}
\newcommand{\sigIneqNames}[1]{{#1}_E}
\newcommand{\sigIneqFam}[1]{{#1}_{\mathrm{ineq}}}

\newcommand{\Term}{\mathrm{Term}}
\newcommand{\var}[1]{\operatorname{var}(#1)}
\newcommand{\constr}[3]{\operatorname{constr}_{#1}(#2, #3)}

\newcommand{\Ineq}{\mathrm{Ineq}}

\newcommand{\varAssignment}{\rho}
\newcommand{\formalIneqVar}[3]{{#1} \mathop{\widetilde{\sqsubseteq}_{#2}} {#3}}
\newcommand{\formalIneq}[2]{{#1} \mathop{\widetilde{\sqsubseteq}} {#2}}

\newcommand{\interpretation}[2]{\operatorname{op}_{#1,#2}}
\newcommand{\PreAlg}[1]{\mathrm{PreAlg}_{#1}}
\newcommand{\Alg}[1]{\mathrm{Alg}_{#1}}

\newcommand{\Initial}[1]{\mathrm{Initial}_{#1}}
\newcommand{\Leq}[1]{\sqsubseteq_{#1}}
\newcommand{\app}[1]{\operatorname{app}_{#1}}
\newcommand{\initialLub}[1]{\textstyle\bigsqcup_{#1}}
\newcommand{\rec}{\mathrm{rec}}
\newcommand{\ind}{\mathrm{ind}}

\newcommand{\sigOver}[2]{#1 + #2}
\newcommand{\inclName}{\tilde{\iota}}
\newcommand{\forgetIncl}[1]{{#1}^{-}}
\newcommand{\addIncl}[2]{{{#1} + {#2}}}

\newcommand{\Universe}{\mathcal{U}}

\newcommand{\botName}{\tilde{\bot}}
\newcommand{\botTerm}{\bot_t}

\newcommand{\pointedDcpoSig}{\Sigma_{\mathrm{DCPO}_\bot}}

\newcommand{\plotkinPower}[1]{\mathcal{P}(#1)}
\newcommand{\formalUnion}{\cup}
\newcommand{\powerIncl}[1]{\{{#1}\}}
\newcommand{\inclusionName}{\widetilde{\powerIncl{\blank}}}
\newcommand{\formalUnionName}{\widetilde{\formalUnion}}
\newcommand{\powerSig}{\Sigma_{\text{Power}}}

\newcommand{\csumSig}[2]{\Sigma_{{#1} + {#2}}}
\newcommand{\sumInclName}[1]{\tilde{\iota}_{#1}}
\newcommand{\sumIncl}[1]{\iota_{#1}}

\newcommand{\smashSig}[2]{\Sigma_{{#1} \wedge {#2}}}
\newcommand{\smashIncl}{\iota}

\newcommand{\coeqSig}[2]{\Sigma_{{#1} \rightrightarrows {#2}}}
\newcommand{\coeqIncl}{\iota}

\newcommand{\emptyList}{[]}
\newcommand{\cons}[2]{#1 :: #2}
\newcommand{\consSorted}[3]{#1 \mathop{{::}\langle#2\rangle} #3}
\newcommand{\finSet}{\mathrm{FinSet}}
\newcommand{\sortedList}{\mathrm{SortedList}}
\newcommand{\strictSortedList}{\mathrm{StrictSortedList}}
\newcommand{\Nat}{\mathbb{N}}
\newcommand{\leqList}{\leq_L}

\newcommand{\blank}{{-}}
\newcommand{\bang}{{!}}